\documentclass[journal]{IEEEtran}
\usepackage{color}
\usepackage{ifpdf}

\usepackage{cite}

\ifCLASSINFOpdf
  \usepackage[pdftex]{graphicx}
\else
\fi
\usepackage{graphicx}
\usepackage{epstopdf}
\usepackage{subfigure}

\usepackage{amsmath}
\usepackage{amsfonts}
\usepackage{amssymb}
\ifCLASSOPTIONcompsoc
  \usepackage[caption=false,font=normalsize,labelfont=sf,textfont=sf]{subfig}
\else
  \usepackage[caption=false,font=footnotesize]{subfig}
\fi
\usepackage{fixltx2e}
\usepackage{footnote}
\makesavenoteenv{tabular}

\ifCLASSOPTIONcaptionsoff
  \usepackage[nomarkers]{endfloat}
 \let\MYoriglatexcaption\caption
 \renewcommand{\caption}[2][\relax]{\MYoriglatexcaption[#2]{#2}}
\fi
\usepackage{comment}
\usepackage{amsthm}
\usepackage{ifthen}

\newtheorem{remark}{Remark}
\newtheorem{proposition}{Proposition}
\newtheorem{definition}{Definition}

\newboolean{showcomments}
\setboolean{showcomments}{true}
\newcommand{\wang}[1]{\ifthenelse{\boolean{showcomments}}
	{ \textcolor{red}{(ZW:  #1)}}{}}
\newcommand{\fliu}[1]{\ifthenelse{\boolean{showcomments}}
	{ \textcolor{red}{(FL:  #1)}}{}}
\newcommand{\zhao}[1]{\ifthenelse{\boolean{showcomments}}
	{ \textcolor{blue}{(CZ:  #1)}}{}}
\newcommand{\slow}[1]{\ifthenelse{\boolean{showcomments}}
	{ \textcolor{blue}{(SL:  #1)}}{}}

\begin{document}
	
	\title{Towards High-Efficiency Cascading Outage Simulation and Analysis in Power Systems: A Sequential Importance Sampling Approach}
	
	\author{Jinpeng~Guo,~\IEEEmembership{Student Member,~IEEE, }
		Feng~Liu,~\IEEEmembership{Member,~IEEE, }
        Jianhui~Wang,~\IEEEmembership{Senior Member,~IEEE, }
		Junhao~Lin,~\IEEEmembership{Student Member,~IEEE, }
		and~Shengwei~ Mei,~\IEEEmembership{Fellow,~IEEE}
		\thanks{Manuscript received XXX, XXXX; revised XXX, XXX. \textit{(Corresponding author: Feng Liu)}.}
       	\thanks{Jinpeng Guo, Feng Liu, and Shengwei Mei are with the Department of Electrical Engineering, Tsinghua University, Beijing, 100084, China  (e-mail: lfeng@tsinghua.edu.cn). } 
       	\thanks{Jianhui Wang is with the the Energy Systems Division, Argonne National Laboratory, Argonne, IL 60439 USA (e-mail: jianhui.wang@anl.gov)}
		\thanks{Junhao Lin  is with the Department of Electrical and Electronic
			Engineering, University of Hong Kong, HKSAR, Hong Kong (e-mail: jhlin@eee.hku.hk).}

        }

	\markboth{REPLACE THIS LINE WITH YOUR PAPER IDENTIFICATION NUMBER (DOUBLE-CLICK HERE TO EDIT}%
	{Shell \MakeLowercase{\textit{et al.}}:  Bare Demo of IEEEtran.cls for Journals}
	
	\maketitle
	
	\begin{abstract}
This paper addresses how to improve
the computational efficiency and estimation reliability in cascading outage analysis. We first formulate a cascading outage as a Markov chain with specific state space and transition probability by leveraging the Markov property of cascading outages. It provides a rigorous formulation that allows analytic investigation on cascading outages in the framework of standard mathematical statistics. Then we derive a sequential importance sampling (SIS) based simulation strategy for cascading outage simulation and blackout risk analysis with theoretical justification. Numerical experiments  manifest that the proposed SIS strategy can significantly bring down the number of simulations and  reduce the estimation variance of cascading outage analysis compared with the traditional Monte Carlo simulation strategy.
\end{abstract}
	

	\IEEEpeerreviewmaketitle

	\section{Introduction}
	\IEEEPARstart{C}{ascading} outage is a sequence of component outages triggered by one or several initial disturbances or failures of system components \cite{r1,r2}. In certain extreme conditions, cascading outages can lead to unacceptably serious consequences. A number of blackouts happened in the power systems worldwide in recent years is a case in point  \cite{r4,r5}. Despite that the probability of blackouts due to cascading outage is tiny, the catastrophic consequence and vast influential range raise great attention to the investigation of cascading outages, especially  in  large-scale interconnected power systems.
	
	Due to the random nature of cascading outages, statistics and probability analysis are extensively deployed as basic mathematic tools to analyze cascading outages based on historical data \cite{r6,r7,r8}. However, it is difficult to acquire accurate and adequate data in practice as blackouts are essentially rare events and very limited information has been recorded to date. In this regard, several high-level statistic models were proposed for analyzing cascading outages, such as CASCADE model \cite{r9} and branching process model \cite{r10,r11}. These models aim to capture the macroscopic features of the overall system in the sense of statistics while omitting the details of cascading outage process. To achieve a closer sight into  cascading outages, researchers, however, consider to pick up such details back, including the uncertain occurrence of initial disturbances, action of protection and dispatch of control center. This consequentially results in different blackout models, such as hidden failure model \cite{r12,r13}, ORNL-PSerc-Alaska (OPA) model \cite{r14,r15}, AC-OPA model \cite{r16}, to name a few. As this kind of approaches are capable of analyzing the cascading outage process in a detailed way, it is expected to exploit mechanisms behind cascading outages by carrying massive simulations  on these models.
	
	Regarding every simulation as an independent identical distribution (i.i.d.) sample, the simulation-based cascading outage analysis is essentially a statistic analysis based on a sample set produced by Monte Carlo simulations (MCS). In the past decade, the MCS approach contributes a lot to reveal the underlying physical mechanism of cascading outages and has been popularly used. However, the intrinsic deficiency of the MCS seriously limits its practicability and  deployments. The main obstacle stems from the notorious "curse of computational dimensionality". It is recognized that a realistic large-scale power system is always composed of numerous components, such as transmission lines,  transformers, and  generators. The possible evolutionary paths of cascading outages diverge dramatically. Hence a specific cascading outage with serous consequence is indeed an extremely rare event. In this context, the MCS analysis turns out to be computationally intractable as a huge number of simulations are required to achieve a reliable estimation of the probability distribution of cascading outages. Empirical results also confirm that the estimation variance can remain unacceptably large even if thousands of simulations have been conducted for a system with only tens of buses. This crucial issue, however, has not been cared seriously enough in the literature and the reliability of the MCS-based blackout analysis could be overestimated to a large extent. This motivates two essential questions: 1) how many simulations are required to guarantee the reliability of the estimation? and 2) whether or not the number of simulations can be effectively reduced without degrading the reliability of the estimation?
	
	 	 In \cite{r17},  a condition is proposed to characterize the relationship between the estimation accuracy and the sample size of the MCS, answering the first question. As for the second question is still kept open to date. This paper aims to bridge the gap both theoretically and algorithmically. Noting  that the theoretic results given in \cite{r17} are built on the standard MCS, it is intuitive to expect that the sample size could be shrunken by adopting certain advanced sampling techniques instead of the naive Monte Carol sampling strategy. In the literature, importance sampling (IS) is an effective method to improve the efficiency of the MCS \cite{r18,r19}, which has already been successfully deployed in various fields including power systems, such as security analysis for power grids \cite{r20} and  risk management in electricity market \cite{r21}. It has also been intuitively used in cascading outage simulations in a heuristic manner \cite{r28,r29,r30,r31}.  Nevertheless, due to the absence of solid mathematical formulation of cascading outages, it is difficult to carry on the analytic investigation in a rigorous fashion. It is also unknown what are the  scope and the conditions of the application of the IS strategy in cascading outage simulations, and how to set the parameters of the IS in simulations.
	 	 	
	Sequential importance sampling (SIS) is an extension of the IS method, which decomposes the  IS into a sequence of sampling steps to facilitate  the implementation for multi-period random process \cite{r18, r22}. Inspired by the success of IS/SIS in diverse fields, this paper applies the SIS to derive a novel simulation strategy for the sake of achieving a high-efficiency and reliable cascading outage analysis. The main contributions of this paper are twofold:
	
	\begin{enumerate}
		\item 	The process of cascading outage in power systems is formulated as a Markov chain. Differing from the current formulations in the literature, we specifically define the state space and transition probability associated with the Markov chain, resulting a well-defined analytic model of cascading outages. Based on the formulation, rigorous mathematical statistics analysis is allowed.
		\item 	Benefiting from the proposed analytic formulation, a high-efficiency cascading outage simulation strategy is derived based on the SIS with theoretical guarantees. Taking full advantage of the Markov property of cascading outages, it is capable of considerably reducing both the number of simulations and the estimation variance.
	\end{enumerate}

	We demonstrate the proposed formulation and simulation strategy outperforms the traditional MCS strategy using the data of standard IEEE 300-bus system and a real provincial power grid   in China.The rest of this paper is organized as follows: traditional blackout modeling and the MCS based analysis are briefly reviewed   in Section II. Section III gives the new formulation of cascading outages based on Markov chain. Then the SIS based simulation strategy associated with the theoretic analysis are presented in Section IV. Case studies in Section V show the benefits and efficiency of the proposed simulation strategy. Finally,  Section VI concludes the paper with remarks.

	\section{ MCS-based Analysis on Cascading Outages}
    In cascading outage analysis, load shedding is usually adopted to evaluate the severity of the cascading outage. To characterize load shedding distribution as a consequence of cascading outages, various kinds of blackout models are built by emulating the cascading outage process in a "{\em descriptive}" way. Massive simulations on such models can provide a number of i.i.d. samples for statistical analyses. This approach is essentially based on the MCS if one simply regards each simulation as a sample. Though there are many kinds of blackout models, the simulation principles are quite similar. In simple terms, at the $j$-th stage of the $i$-th sampling of a cascading outage, the blackout model for simulation determines the outage probability of each component in the system at the next stage, depending on the system states $x_j^i$ and other related factors, such as weather and maintenance conditions. Then the outage components at stage $j$ are sampled and the system state $x_j^i$ transits to ${x_{j+ 1}^{i }}$. Repeating the above steps until there is no occurrence of new outages , one simulation is completed. It gives a sample of load shedding $Y$\footnote{Here, load shedding, $Y$, is recognized as a random variable, as we will strictly define later on.}, denoted by $y_M^i$.

    Define the sample sets as ${Y_M}: = \{ y_M^i,i = 1, \cdot  \cdot  \cdot {N_M}\} $ obtained after $N_M$ simulations. Then the probability distribution of load shedding can be estimated by statistics based on $Y_M$. We care about the probability of a given incident $A$ that describes the load shedding being greater than a certain level $Y_{0} $. The unbiased estimation of the true probability $\mu (A)$ is given by
	\begin{equation} \label{eq1}
	\tilde{\mu} (A)=\frac{1}{N_M} \sum \nolimits_{i=1}^{N_M}\delta _{\{ {y_M^i} \ge Y_{0} \} }
	\end{equation}
	where $\delta _{\{ \cdot \} } $ is the indicator function of set $\{ y^{i}_M \ge Y_{0} \} $, which means $\delta _{\{ y^{i}_M \ge Y_{0} \} } =1$ if $ y_{i} \ge Y_{0}$; otherwise,  $\delta _{\{ y^{i}_M \ge Y_{0} \} } =0$. It is easy to see ${\delta _{\{  \cdot \} }} = \delta _{\{  \cdot \} }^2$.
	
	The  variance of the estimation on $N_M$ samples is given by
	\begin{equation} \label{eq2}
	\sigma ^2(A) = D(A)={\frac{1}{N_M}\left( \mu (A)(1-\mu (A)) \right)}
	\end{equation}
	
	In addition to probability distribution of the load shedding, another important indicator in the cascading outage analysis is the blackout risk. Theoretically, the blackout risk of a power system can be defined as the expectation of such load shedding greater than the given level, $Y_0$. That is
	\begin{equation} \label{eq3}
	Risk(Y_0)=\mathbb{E} (Y \cdot \delta _{\{ Y \ge Y_{0} \} } )
	\end{equation}
	Similar to the probability estimation, it can be estimated by
	\begin{equation} \label{eq4}
	\tilde Ris{k}(Y_0) =\frac{1}{{{N_M}}}\sum\nolimits_{i = 1}^{{N_M}} {y_M^i \times {\delta _{\{ y_M^i \ge {Y_0}\} }}}
	\end{equation}
	
	The definition of the blackout risk in \eqref{eq3} represents the risk of cascading outages with serious consequences. It is  closely related to the well-known risk measures, value at risk (\textit{VaR}) and conditional value at risk (\textit{CVaR}) \cite{r23,r24}.  Actually, the  risk measure, \textit{Risk}, defined in \eqref{eq3} is \textit{CVaR}$_\alpha$ times $(1-\alpha )$, provided  \textit{VaR}$_\alpha$ is known as the risk associated to the given load shedding level $Y_{0}$ with a confidence level of $\alpha$.
	
	With the sample set obtained by repeatedly carrying out simulations on the cascading outage model, the probability of load shedding and the blackout risk can be estimated by using \eqref{eq1} and \eqref{eq4}, respectively. However, it should be noted that, to achieve acceptably small variance of estimation, a tremendous number of samples are usually required even if the system merely has tens or hundreds of buses. To illustrate this matter of fact, we use the IEEE 300-bus system as an example. Based on OPA model, the  probability of load shedding is estimated by using \eqref{eq1} based on 10 groups of simulations, where each group contains 2000 i.i.d. simulations. As shown in Fig. \ref{fig.1}, the variance of the probability estimation is quite large. It is also found that the simulations capture few events with the load shedding larger than 800MW, showing that the traditional MCS approach might be neither efficient nor reliable enough to cope with the cascading outage analysis in large-scale power systems. To the best of our knowledge, this issue has not been paid enough attention in the literature.
		\begin{figure}[htb]
			\centering
           \includegraphics[width=0.4\textwidth]{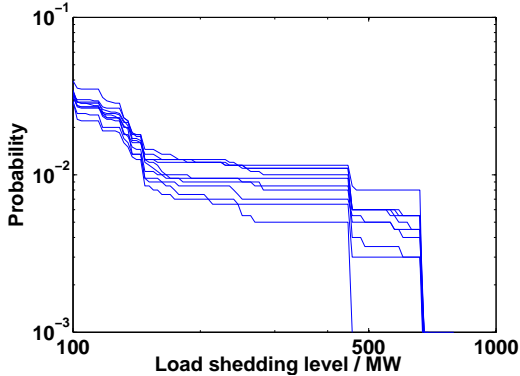}
			\caption{Probability  estimation of the load shedding}
			\label{fig.1}
		\end{figure}
\section{ A Markov Chain Based Formulation}
	A cascading outage is always triggered by one or several initial disturbances or componentwise failures. As a consequence, the protection devices and the control center begin to take actions, and then the system state changes sequentially according to these actions. Such state changes happen in a random way, implying that a cascading outage could be formulated as a  stochastic process. To this end, the state space of the cascading outage as well as the associated state-transition probability have to be defined appropriately. In this paper,  the system configuration is taken as the random state variables. Note that the system configuration defined here is generic, which can incorporate either controlled or uncontrolled changes, such as line tripping, shunt capacitor switching and On-Load Tap Changer regulation. We denote  $X_j$  as the state variable at stage $j$ of a cascading outage.  All possible system states span the state space, denoted by $\mathcal X$.
	
	 Assume the system has $N_c$ components and denote $N:=\{1,2,\cdots n\}(n\in\mathbb{Z^+})$ as the total stages of cascading outages.  Then an $n$-stage cascading outage can be defined below.
	\begin{definition}\label{cascadingoutage}
	An $n$-stage cascading outage is a stochastic sequence $
	Z:= \{ {X_1, X_2, ..., X_j, ..., X_n},  \forall j\in N, X_j\in {\mathcal X} \}
	$
	with respect to the random state space $\mathcal X$ and a given joint probability distribution $f(X_{n} ,\cdots X_{2} ,X_{1} )$.
	\end{definition}
	
	In the above definition, $j$ is the stage label of the cascading outage, while  $n$ is the total number of stages, or the \textit{length} of the cascading outage. State variable $X_{j} $ is a discrete random vector with the dimension of $N_c$.  Each  element of $X_j$ stands for a state of the corresponding component at stage $j$ during the cascading outage. Correspondingly,  $\mathcal X$ is a $N_c$-dimensional state space. Moreover, denoting the number of possible states of component $k$ by $s^k$, there is
	\begin{equation}\label{eq5}
	 |\mathcal X|=  \prod \limits _{k = 1}^{{N_c}} {{s^k}}
	\end{equation}
	where $|\mathcal X|$ denotes the number of elements in $|\mathcal X|$.
	
	For simplicity, we abuse the notation $Z:=\{X_j^N\}$ to denote a cascading outage. Then the joint probability distribution $f(X_{n} ,\cdots X_{2} ,X_{1} )$ is simplified into $f(Z)$. On the other hand, since the number of components in the system is finite, the number of  possible stochastic sequences representing the cascading outages is finite as well. We denote $\mathcal{Z}$ as the set of all possible cascading outages in a system. Thus,  $|\mathcal{Z}|$ is finite.
	
	It is worthy of noting that, the joint probability distribution $f(Z)$ is practically difficult to obtain, even if the probability distribution functions (PDFs) of individual components are known. Next we show this issue can be circumvented by using the intrinsic Markov properties of cascading outages.	
	
	In Definition 1, for a given $n$-stage cascading outage, the associated load shedding is merely a stochastic variable being a function of the stochastic sequence  $\{X_j^N\}$, which is denoted by  $Y=h(X_{1} ,\cdots, X_{n} )=h(Z)$ .

	Note that in a cascading outage process, all the actions of protections, controls and operations at arbitrary stage $i$  are completely determined by the previous stage $i-1$.  In this context, the cascading outage $\{ X_j^N\} $ in the definition above is indeed a Markov chain. Then invoking the conditional probability formula and the Markov property, the joint probability distribution $f(Z)$ should satisfy
		\begin{equation} \label{eq6}
		\begin{array}{rcl}
		f(Z) & = &f({X_n}, \cdots, {X_2},{X_1}) \\
		&= &f_n({X_n}|{X_{n - 1}} \cdot  \cdot  \cdot {X_1})\cdot f_{n-1}({X_{n - 1}}|{X_{n - 2}} \cdots {X_1}) \\
		 & & \cdots f_2(X_2|X_1)\cdot f_1({X_1}) \\
		&= &f_n({X_n}|{X_{n - 1}})\cdot f_{n-1}({X_{n - 1}}|{X_{n - 2}}) \cdot  \cdot  \cdot f_1({X_1})
		\end{array}
		\end{equation}
	where $f_{j+1}({X_{j+1}}|{X_j})$ is the related conditional probability.	

	Assume in the sampling process, $x^i_j$ is the sample of the state at stage $j$ of the $i$-th sampling, while the length of the cascading outage in the $i$-th sampling is $n^i$.  Then  we have
	\begin{equation} \label{eq7}
	\begin{aligned}
	& {\mathbf{Pr}({X_{n^i}} = {x^i_n}|{X_{{n^i} - 1}} = {x^i_{{n^i} - 1}}, \cdots {X_2} = {x^i_2},{X_1} = {x^i_1})} \\
	 =&{ \mathbf{Pr}({X_{n^i}} = {x^i_n}|{X_{{n^i} - 1}} = {x^i_{{n^i} - 1}})} \\
	\end{aligned}
	\end{equation}

	 Eqs. \eqref{eq6} and \eqref{eq7} mathematically indicate that,  a cascading outage can be simulated following the sequential conditional probability, other than directly using the joint probability distribution. Specifically, denote ${F^i_{j}} $ as the set of outage  components at stage $j$ of the $i$-th  sampling of the cascading outage,  and $\bar {F}^i_{j} $ as the set of the normal components after stage $j$ of the $i$-th sampling. Let
		\begin{equation} \label{eq8}
            p_{j,k}^i = {\varphi _k}(x_j^i)
		\end{equation}
as the outage probability of component $k$ at stage $j$ of the $i$-th sampling, where $\varphi _k$ is the corresponding PDF.  Then the transition probability, $\hat{p}^i_{j,j+1}$, from state ${x^i_{j}} $ to state ${x^i_{j + 1}} $  is
	  \begin{equation} \label{eq9}
	  {\hat{p}^{i}_{j,j + 1}} = f({x^i_{j + 1}}|{x^i_{j}}) = \prod\limits_{k \in {F^i_{j}}} {p^i_{j,k}} \prod\limits_{k \in {\bar{F}^i_{j}}} {(1 - p^i_{j,k})}
	  \end{equation}

	Based on \eqref{eq9}, the probability of the $i$-th  sample of the cascading outage (the complete path), denoted by $p_c^i$, is given by 	
	\begin{equation} \label{eq10}
	{p_c^{i}} = \prod\limits_{j = 1}^{n^i-1} {{\hat{p}^i_{j,j + 1}}}
	\end{equation}
	  \begin{remark}
	  	\label{R1}
	  	   This sequential treatment actually has been heuristically  used in most cascading outage simulations albeit without justifying its validity. By strictly defining cascading outages as a Markov chain with appropriate state space and transition probability distribution, our work provides not only a justification for such a extensively-used treatment, bust also a solid mathematical foundation for deriving efficient cascading outage simulation strategies and carry out theoretical analysis, as we discuss in Section IV.
	  \end{remark}
	  \begin{remark}
	  	  	\label{R2}
			Eq. \eqref{eq10}  indicates that the probability of a cascading outage can be very small as it is  the product of a series of small probabilities. Particularly, a cascading outage with a severe consequence usually involves many stages with very small probabilities, resulting in an extremely small probability. It is the main cause that blackout events can hardly be captured by using traditional MCS. As a consequence, insufficient samples of rare events may further give unreliable estimation results of the blackout risk with  biased expectation and/or large variance. This problem, theoretically, cannot be alleviated effectively in large-scale system by merely increasing the number of simulations, as the size of state space $\mathcal X$  expands dramatically when the number of system components increases (according to Eq. \eqref{eq5}).
	  	 \end{remark}

	\section {Cascading Outage Simulation Based on SIS}
	
	\subsection{Importance Sampling for Cascading Outage Simulations}
	For improving the sampling efficiency and depressing the estimation variance, importance sampling (IS) technique is recognized an effective tool. Its basic idea is to sample the stochastic process under a proposal joint probability distribution $g(X_{n} ,X_{n-1} ,\cdots , X_{1} )$  ($g(Z)$ for short) instead of the true joint probability distribution $f(Z)$. Specially, the probability of arbitrary possible cascading outage under the proposal joint probability distribution needs to be positive, i.e., $g(Z)>0,Z \in \mathcal Z$. Then after $N_{s} $  i.i.d. simulations, we can obtain a sample set of cascading outages, $Z_s:=\{z_s^i, i =1, 2, \cdots, N_s\}$, where, ${z_s^i} = \{{x^i_{1}}, x^i_2, \cdots, {x^i_{n^i}}\} $ is the $i^{th}$ sample of cascading outages;  $n^i$  the length of the sampled cascading outage in the $i^{th}$  simulation; 	
	 $x^i_j$  the sampled state at stage $j$ of the $i^{th}$ simulation. Afterward, we can  obtain the sample set of load shedding, $Y_s:=\{ y^i_{s}, i=1\cdots N_{s} \} $, where $y^i_{s}=h(z^i_{s})$. For simplicity, we abuse the notation $\delta_{Y_0}$ throughout to stand for $\delta_{\{ h(Z) \ge Y_0\}}$. As $|\mathcal Z|$ is finite, the true probability of event $A$ defined  previously is given by
		\begin{equation} \label{eq11}
		\mu (A) =\sum\limits_{Z \in \mathcal Z} {{\delta _{ Y_0}}f(Z)}
		\end{equation}

	As the true probability $\mu(A)$ cannot be obtained accurately, we usually estimate it through Eq. \eqref{eq1} based on $N_M$ samples given by the MCS under the original joint probability distribution $f(Z)$. The variance of estimation, $D(A)$, is given by \eqref{eq2}.
	
	We are interested in the expectation and variance based on the IS under the proposal probability distribution $g(Z)$. To this end, we let $w(Z)=f(Z)/g(Z)$, yielding
		 	\begin{equation} \label{eq12}
		 	\mu (A)=	\sum _{Z \in \mathcal Z} {{\delta _{{Y_0} }}w(Z)} g(Z)
		\end{equation}
	
	As the IS with the proposal joint probability distribution $g(Z)$ is deployed,  the unbiased estimation of $\mu (A)$  based on $N_{s}$ samples turns to be
	\begin{equation} \label{eq13}
	{{\tilde {\mu}} _{IS}(A)=
		\frac{1}{{{N_{s}}}}\left(\sum\limits_{i = 1}^{{N_{s}}} {{w}({z^i_s}) }\cdot {\delta _{\{ {y^i_s} \ge {Y_0}\} }}\right)}
	\end{equation}
	where $w(z^i_s)>0$ is the sampling weight subject to
	\begin{equation} \label{eq14}
	{w}({z_s^i}) = \frac{{f({z^i_s})}}{{g({z^i_s})}}
	\end{equation}
	
	Moreover, the variance of the probability estimation is
	\begin{equation} \label{eq15}
		\begin{array}{ll} {{D _{IS}}(A)} &{ = \frac{{\mathbb{E}{{\{ {\delta _{{Y_0} }}w(Z) - \mathbb{E}[{\delta _{{Y_0} }}w(Z)]\} }^2}}}{{{N_{s}}}} } \\
		&{ = \frac{{\mathbb{E}\{ {{[{\delta _{{Y_0} }}w(Z)]}^2}\}  - {{\{ \mathbb{E}[{\delta _{{Y_0} }}w(Z)]\} }^2}}}{{{N_{s}}}}}\\
		&{  = \frac{{\sum\limits_{Z \in \mathcal Z} {\delta _{{Y_0} }^2{w^2}(Z)g(Z)}  - {{[\sum\limits_{Z \in \mathcal Z} {{\delta _{{Y_0} }}w(Z)g(Z)} ]}^2}}}{{{N_{s}}}} } \end{array}
	\end{equation}

 Let
	\begin{equation} \label{eq16}
	{w_0} = \frac{{\sum\limits_{Z \in \mathcal Z} {\delta _{{Y_0} }^2{w^2}(Z)g(Z)} }}{{\sum\limits_{Z \in \mathcal Z} {\delta _{{Y_0} }^2w(Z)g(Z)} }}
	\end{equation}
     Then substituting  \eqref{eq12} and   \eqref{eq16} into  \eqref{eq15} yields
    \begin{equation} \label{eq17}
     \begin{array}{ll} {D _{IS} (A)}  &{=\frac{1}{{{N_{s}}}} \left( {{w_0}\sum\limits_{z \in {\mathcal Z}} {{\delta _{{Y_0} }}w(z)g(z)}  - {\mu^2}(A)} \right)} \\ &{=\frac{1}{N_{s} } \left( w_{0} \mu(A)-\mu^{2} (A) \right)} \end{array}
     \end{equation}

Next we present some important propositions.

\begin{proposition}
		\label{p1}
	Given $g(Z)$, $w(Z)$ and $\mathcal Z$, there must be
	 $$ {w_0}  \in [\mathop {min}\limits_{Z \in {\mathcal Z}} w(Z),\mathop {\max }\limits_{Z \in {\mathcal Z}} w(Z)] $$	
\end{proposition}

\begin{proof}
Since $w(Z)$ and $g(Z)$ are non-negative, we have
$$ {w_0} \le \frac{{\sum\limits_{Z \in \mathcal Z} {\delta _{{Y_0} }^2w(Z)g(Z)\mathop {max}\limits_{Z \in \mathcal Z} w(Z)} }}{{\sum\limits_{Z \in \mathcal{Z} } {\delta _{{Y_0} }^2w(Z)g(Z)} }} = \mathop {max}\limits_{Z \in \mathcal Z} w(Z)$$
Similarly, we have
$$ {w_0} \ge \frac{{\sum\limits_{Z \in \mathcal Z} {\delta _{{Y_0} }^2w(Z)g(Z)\mathop {min}\limits_{Z \in \mathcal Z} w(Z)} }}{{\sum\limits_{Z \in \mathcal Z} {\delta _{{Y_0} }^2w(Z)g(Z)} }} = \mathop {\min }\limits_{Z \in \mathcal Z} w(Z) $$
\end{proof}

\begin{proposition}
	\label{p2}
	Let	$D _{IS} (A)$ and $D (A)$ be the variances of the probability estimation of event $A$ defined previously by using the IS and the MCS, respectively. If  $N_s=N_M$, then
	$D _{IS} (A)<D (A)$ holds if and only if the proposal joint probability distribution $g(Z)$ satisfies $w_0<1$, or equivalently,
			\begin{equation} \label{eq18}
			w_{0} \mu(A)<\mu(A)
			\end{equation}
\end{proposition}

\begin{proposition}
	\label{p3}
	Let	$D _{IS} (A)$ and $D (A)$ be the variances of the probability estimation of event $A$ defined previously by using the IS and the MCS, respectively. If $D_{IS}(A)=D(A)$, then
	$N _{IS}<N_M$ holds if and only if the proposal joint probability distribution $g(Z)$ satisfies  $w_0<1$, or equivalently,
	\begin{equation} \label{eq19}
	w_{0} \mu(A)<\mu(A)
	\end{equation}
\end{proposition}
	
 It is easy to prove Proposition 2 and Proposition 3  by directly comparing \eqref{eq3} with \eqref{eq17}.

\begin{remark}
	Proposition 1 guarantees the existence of $w_0$, while Propositions 2 and 3 give the necessary and sufficient conditions that the IS can reduce sample size and estimation variance compared with that obtained by the MCS. In practice, it may be difficult to check the conditions  \eqref{eq18} or  \eqref{eq19}. A more convenient way is to use the following sufficient condition:
	\begin{equation} \label{eq20}
		g(Z)>f(Z), \forall Z\in \{Z \in \mathcal Z|\; h(Z)>Y_0\}
	\end{equation}	
\end{remark}

Similar conclusion can be drawn for the blackout risk assessment. Given $g(Z)$ for the IS, then the blackout risk is
	\begin{equation} \label{eq21}
		Risk_{IS} (Y_0) = \mathbb{E}(Y \cdot w(Z)\cdot \delta_{Y_0})
	\end{equation}	
The estimation of blackout risk based on $N_s$ samples is
	\begin{equation} \label{eq22}
	\tilde Ris{k_{IS}}(Y_0) = \frac{1}{{{N_{s}}}}\left( \sum\limits_{i = 1}^{{N_{s}}} {y_s^iw(y_s^i){\delta _{\{ y_s^i \ge {Y_0}\} }}}\right)	\end{equation}

	Then the estimation  variance of  \eqref{eq4} and \eqref{eq22} are given by
	\begin{equation} \label{eq23}
	D (R) = \frac{{\sum\limits_{Z \in \mathcal Z} {{h^2}(Z){\delta _{{Y_0} }}f(Z)}  - {Risk(Y_0)}^2}}{N_M}
	\end{equation}
and
	\begin{equation} \label{eq24}
	{D _{IS}}(R) = \frac{{\sum\limits_{Z \in \mathcal Z} {{w^2}(Z){h^2}(Z){\delta _{{Y_0} }}g(Z)}  - {Risk(Y_0)}^2}}{{{N_{s}}}}
	\end{equation}
respectively. 	According to \eqref{eq23} and \eqref{eq24}, the condition of the variance reduction can be obtained accordingly.
	
	The theoretical analysis indicates that the IS can reduce both the sample size and the estimation variance, provided an appropriately selected proposal probability distribution $g(Z)$ . Considering the unbiasedness of the estimation using the IS and the MCS, the lower variance indicates that the IS has a better estimation performance than the MCS \cite{r26}.

	\subsection{Sequential Importance Sampling based Simulation Strategy}
	
	Similar to  \eqref{eq6},  for $g(Z)$ we have
	\begin{equation} \label{eq25}
	\begin{array}{ll}
	g(Z)&=g(X_{n} ,\cdot \cdot \cdot X_{2} ,X_{1} )\\
	&{=g_n(X_{n} |X_{n-1} )\cdot g_{n-1}(X_{n-1} |X_{n-2} )\cdot \cdot \cdot g_1(X_{1} )}
	\end{array}
	\end{equation}
	It means the proposal joint probability distribution $g(Z)$ can be chosen sequentially at individual stages in a cascading outage.
	Thus the problem of choosing $g(Z)$ turns out to be one of choosing the series $g_{j+1}(X_{j+1} |X_{j} )$ sequentially. For the purpose of acquiring  more information about the cascading outage with severe load shedding, $g_{j+1}(X_{j+1} |X_{j} )$ should be carefully chosen to amplify the probability of cascading outages in  future stages versus original $f_{j+1}(X_{j+1} |X_{j} )$. Heuristically, we modify the outage probability of components given in \eqref{eq8} into
	\begin{equation} \label{eq26}
	q_{j,k}^i = \min (\eta p_{j,k}^i,\max (\varphi _k))
	\end{equation}
	where $q^i_{j,k}$ is the modified component's outage probability;  $\eta $ is the SIS parameter stands for the amplification factor of component's outage probability. Correspondingly  the modified transition probability becomes
	\begin{equation} \label{eq27}
\hat q_{j,j + 1}^i = \prod\limits_{k \in F_j^i} {q_{j,k}^i} \prod\limits_{k \in \bar F_j^i} {(1 - q_{j,k}^i)}
	\end{equation}
	
	For the $i$-th sample, the original load shedding probability $p_c^i $ is given by  \eqref{eq10} while the modified probability $q_c^i $ is given by
	\begin{equation} \label{eq28}
	q_c^i = \prod\limits_{j = 1}^{{n^i-1}} {\hat q_{j,j + 1}^i}
	\end{equation}
The corresponding sampling weight is
	\begin{equation} \label{eq29}
	w(z_s^i) = \frac{{p_c^i}}{{q_c^i}} = \prod\limits_{j = 1}^{{n^i-1}} {\frac{{\hat p_{j,j + 1}^i}}{{\hat q_{j,j + 1}^i}}}
	\end{equation}
	
	Simulating cascading outages with sampling weights given by \eqref{eq29}, the load shedding probability and the blackout risk can be estimated by using  \eqref{eq13} and \eqref{eq22}, respectively. According to the previous analyses, both the  number of simulations and the variance of estimations can be reduced, provided appropriately selected sampling weights  .

\begin{remark}
To guarantee  high sampling efficiency of the SIS,  $\eta$ should be choose carefully so that \eqref{eq18} or \eqref{eq20} is satisfied. Unfortunately, it is not really a trivial work because, in the light of the necessary and sufficient condition\eqref{eq18},  $w_0$  cannot be known a priori. However, noticing $p^i_{j,k}$ in \eqref{eq9} and and $q^i_{j,k}$ in \eqref{eq29} are usually very small, we have $(1-p^i_{j,k}) \approx (1-q^i_{j,k}) \approx 1$. It implies that the following condition holds
$$w(z_s^i) = \prod\limits_{j = 1}^{{n^i}} {\frac{{\hat p_{j,j + 1}^i}}{{\hat q_{j,j + 1}^i}}}  \approx \prod\limits_{j = 1}^{{n^i}} {\frac{{\prod\nolimits_{k \in F_j^i} {p_{j,k}^i} }}{{\prod\nolimits_{k \in F_j^i} {q_{j,k}^i} }}} $$
for most samples. Thus, if $\eta$ is selected such that $\eta>1$,  then the sufficient condition \eqref{eq20} can  hold approximately. Numerical experiments  empirically support this conclusion.

\end{remark}

	\subsection{Algorithm}	
	
	The algorithm of the SIS based strategy is given as follows
	
\noindent\rule[0.25\baselineskip]{0.5\textwidth}{1pt}	
	\begin{itemize}
		\small
		\item {\bf Step 1:  Data preparation.} Initialize the system data and parameters. Specifically, choose $\eta>1$.
		\item {\bf Step 2: Sampling states.} For the $i$-th sampling, according to the system state, $ x_j^i $ at stage $j$ , and the outage probability of components based on \eqref{eq8} and \eqref{eq26}, simulate the  component outages and acquire the new state $ x_{j+1}^i$ at the next stage. Afterward, calculate the state transition probability and the sampling weight using \eqref{eq10} and \eqref{eq28}, respectively.
		\item {\bf Step 3:  Termination judgment.} If $x_j^i $ is the same as $x_{j+1}^i$, the $i$-th sample of cascading outage simulation is completed at stage $j$ and the $i$-th sample $z_s^i = \{ x_1^i \cdot  \cdot  \cdot x_j^i\}  $ is obtained. If all $N_{s} $ simulations are completed, the sampling process is ended; otherwise let $i=i+1$ and go back to Step \eqref{eq2}.
		\item {\bf Step 4:  Data analysis.} According to  \eqref{eq13} and \eqref{eq22}, estimate the probability of load shedding and blackout risk.
	\end{itemize}
\noindent\rule[0.25\baselineskip]{0.5\textwidth}{1pt}

\begin{remark}
 
	In addition to the IS/SIS, the SPLITTING method has been used for effectively improving the rare events analysis in power systems \cite{r32,r33,r34}. Its main idea is to divide the path of cascading outages into multiple sub-paths to dramatically increase the  probability of rare events of interest. Similar to the IS/SIS, its simulation settings and parameters must be tuned carefully. 
	As the SPLITTING is still a MCS-based method essentially, the simulations for each sub-path may still need a huge number of samples as the state space is large. It is interesting that this problem can be surmounted by using the IS/SIS. This further motivates an improved approach that combines both the IS/SIS and the SPLITTING methods, which is our ongoing work.     

\end{remark}

	\section{Case Studies}
	In this section, numerical experiments are carried out on two systems based on the simplified OPA model without slow dynamic \cite{r14}. One test system is the IEEE 300-bus system with a total load of $24,000$ MW, while the other is a real provincial power grid  in China, with $1,122$ buses ,$1,792$ transmission lines or transformers and $52,000$ MW total load. 
	
	\subsection {Case 1: IEEE 300-bus System}

	\subsubsection { Efficiency of Probability Distribution Estimation}	
	In this case, the probability of load shedding in IEEE 300-bus system  is estimated by using the MCS and the SIS, respectively. The sample size of the MCS is 50,000 while that of the SIS is only 2,000 as the MCS requires much more samples to achieve a small variance of estimation.
	As mentioned previously, both the MCS and the SIS strategies give unbiased estimation on the load shedding probability. According to the estimation results shown in Fig. \ref{fig.2},  the two strategies output almost the same estimations on the probability distribution as the load shedding  less than 1,000MW. This result justifies that the SIS simulation strategy can achieve a given estimation accuracy with much less number of simulations, and thus it is of  higher efficiency than the MCS strategy. 	
		\begin{figure}[htb]
			\centering
			\includegraphics[height=4.0cm,width=6.5cm]{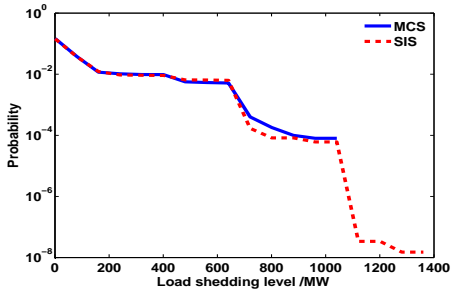}
			\caption{Probability estimation of the load shedding with MCS and SIS}
			\label{fig.2}
		\end{figure}
	
	In terms of the load shedding greater than 1,000MW (the corresponding probability is less than  $10^{-4}$ ), the MCS fails to find any event in 50,000 simulations and cannot come up to estimation for such very rare events. In the contrary, the SIS strategy successfully finds out many rare events with load shedding as large as 1,400MW in only 2,000  simulations (the corresponding probability is nearly $10^{-8}$  ). It indicates that the SIS strategy can considerably facilitate capturing very rare events of cascading outages even with much less  simulations. It also implies that the blackout risk analysis based on the the MCS might not be reliable enough since the captured rare events are usually far from being sufficient.
	
	\subsubsection {Variance of  Probability Distribution Estimation}	
	In this case, we compare the variance of probability estimation with the two strategies (see Fig. \ref{fig.3}) in IEEE 300-bus system. Since the true variance of probability estimation cannot be obtained directly, the sample variance is used as a surrogate. Take the MCS as an example. Denote $\tilde{\mu}^m (A)$ as the estimation of $m$-th sample sets, then the sample variance is
	$\tilde D(A)=\frac{1}{{{m_{\max }} - 1}}\sum\limits_{m = 1}^{{m_{\max }}} {{{[{\mu^m}(A) - (\frac{1}{{{m_{\max }}}}\sum\limits_{m = 1}^{{m_{\max }}} {{\mu^m}(A)} )]^2}}}$, where ${m_{\max }}$ is the number of i.i.d sample sets, which is set as 75 here.

	For comparison, the sample sizes of the MCS and the SIS are both set as 2,000. The SIS parameter is selected as $\eta=1.5$.
	As shown in Fig. \ref{fig.3},  the estimation variance for the SIS is lower than the MCS. The equivalent sampling weight bound $w_0P(A)$ is given in Fig \ref{fig.4}. It shows that the sufficient condition \eqref{eq20} is  satisfied almost everywhere, empirically verifying the theoretic analysis in Remark 4.
				\begin{figure}[htbp]
					\centering
					\includegraphics[height=4.0cm,width=6.5cm]{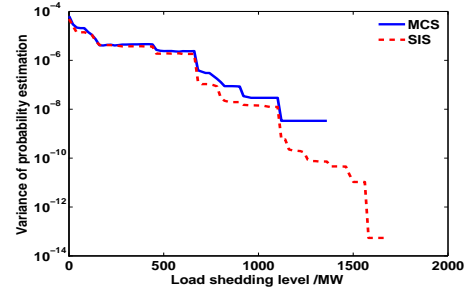}
					\caption{Variance of probability estimations with MCS and SIS}
					\label{fig.3}
				\end{figure}	
			\begin{figure}[htbp]
			\centering
			\includegraphics[height=4.0cm, width=6.5cm]{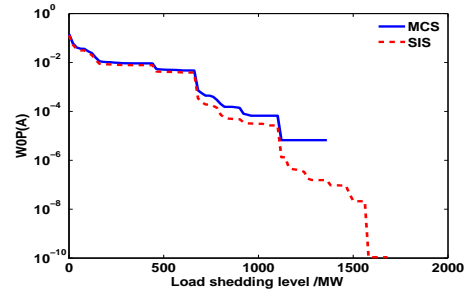}
			\caption{$w_0P(A)$ v.s. $P(A)$}
			\label{fig.4}
		\end{figure}

Fig. \ref{fig.5} presents the estimation variances decrease along with the increase of sample size. Here, the probability is estimated according to cascading outages with load shedding larger than (a)650MW, (b)750MW and  (c)850MW, respectively. As shown in Fig. \ref{fig.5}, the estimation variances of the SIS decrease much  faster compared with that of the MCS,  demonstrating that  SIS simulation strategy is capable of achieving more reliable estimation results  with much less simulations.
	\begin{figure}[htbp]
		\centering
		\includegraphics[height=7.6cm, width=0.5\textwidth]{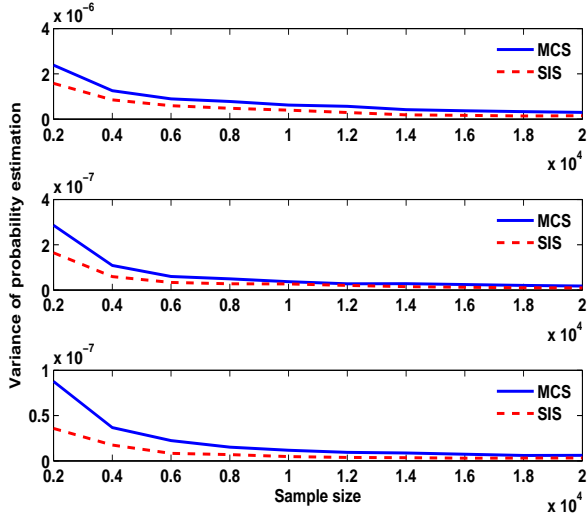}
		\caption{Convergence of the variance of the probability estimation}
		\label{fig.5}
	\end{figure}
\subsubsection {Impacts of the SIS Parameters $\eta$}
	In this case, we analyze the influence of the SIS parameter $\eta$ on the estimation variance of cascading outages in IEEE 300-bus system (see Fig. \ref{fig.6}).  Here,  $\eta$ is selected as 1.2, 1.5, 2, respectively, while other conditions are the same as the previous cases. It is found that $\eta$ can impact on the probability estimation of cascading outages in twofold: 	On the one hand, as a larger $\eta$ is adopted, more detailed information of the rare events can be captured. From Fig. \ref{fig.6}, it is observed that the SIS with $\eta=2$ obtains blackout samples with load  shedding even over 2,000MW (the corresponding probability is nearly $10^{-16}$), while the SIS with a smaller $\eta$, say 1.2 or 1.5, does not capture such rare events.
	\begin{figure}[htb]
	\centering
	\includegraphics[height=5.6cm, width=0.53\textwidth]{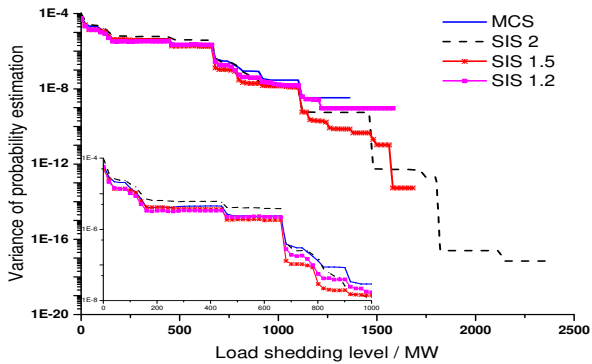}
	\caption{ Variance of probability estimation with different SIS parameters}
	\label{fig.6}
	\end{figure}
					
On the other hand, whereas more rare event samples are captured, the estimation variance of normal events with lower load shedding increases. In this case, the SIS with $\eta=2$ exhibits larger variance of the probability estimation of the load shedding less than 600MW versus  either the SIS with smaller parameters or the MCS. However, when $\eta$ is scaled down to 1.5 or 1.2, the variance of the probability estimation of normal events drops down to the same as the MCS, albeit much less rare events can be found. This case empirically indicates,  a larger SIS parameter can facilitate capturing more rare events with higher load shedding, at the expense of increasing the estimation variance of normal events. This expense, nevertheless, does make sense and is acceptable as we mainly are concerned with the potential blackouts with quite large load shedding. This feature of the SIS also allows to  purposely adjust resolution of cascading outage analysis according to desired levels of load shedding by carefully tuning the SIS parameter.

\subsubsection{Blackout Risk Estimation}	
	In this part, we deploy the SIS and the MCS simulation strategies to analyze the blackout risk defined as in \eqref{eq4}, where the load shedding level $Y_0$ is set as 750MW.  The mean value and the variance are obtained based on 75 sample sets. In each of sample set, 2,000  simulations are carried out with the SIS and the MCS, separately. The curves of the mean value and the variance along with the sample size are shown in Fig. \ref{fig.7}, showing that the SIS can significantly improve both the efficiency and the reliability of  blackout risk analysis.

\begin{figure}
	\centering
	\subfigure[Mean value]{
		\includegraphics[ width=0.48\columnwidth]{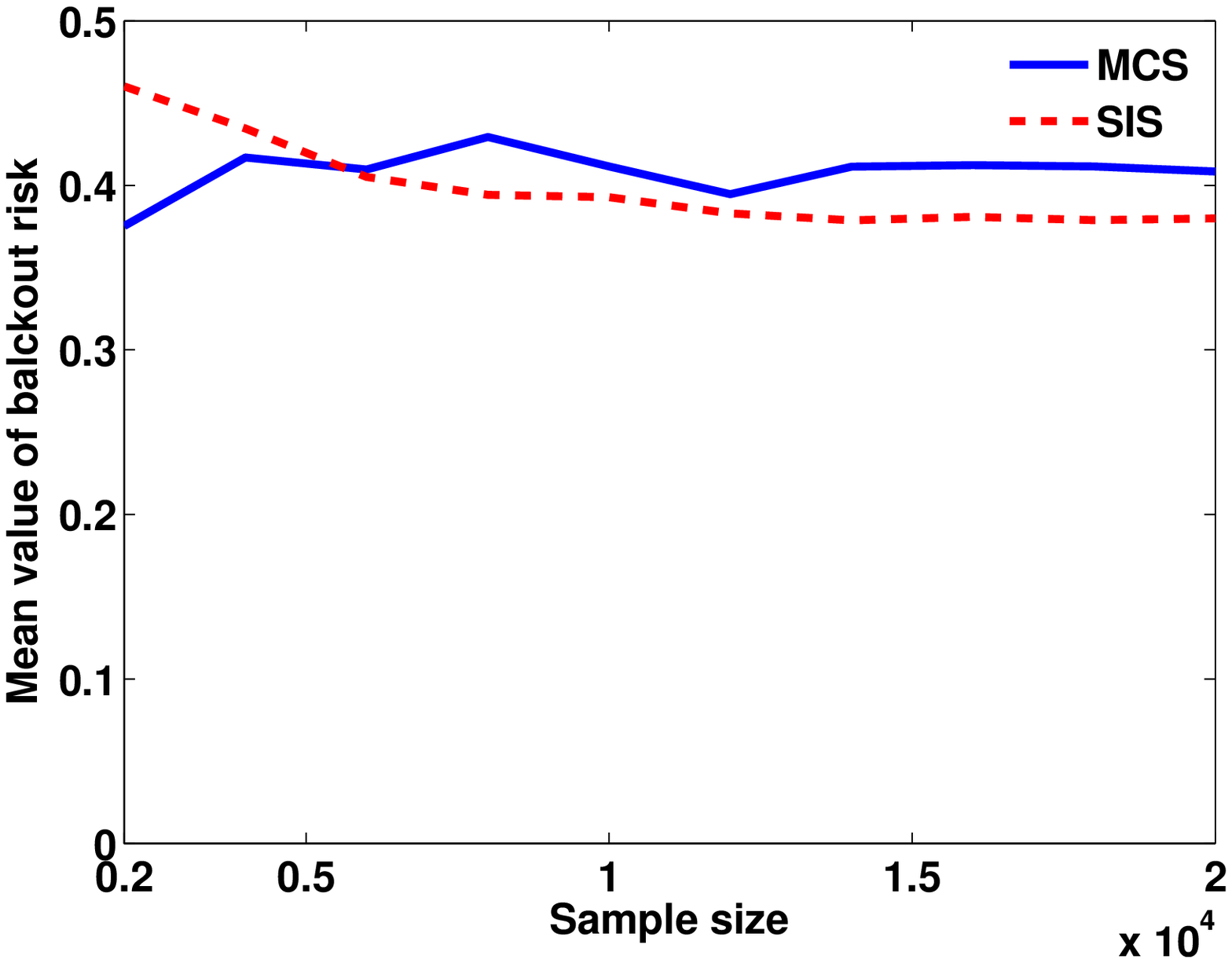}}
	\subfigure[Variance]{
		\includegraphics[ width=0.48\columnwidth]{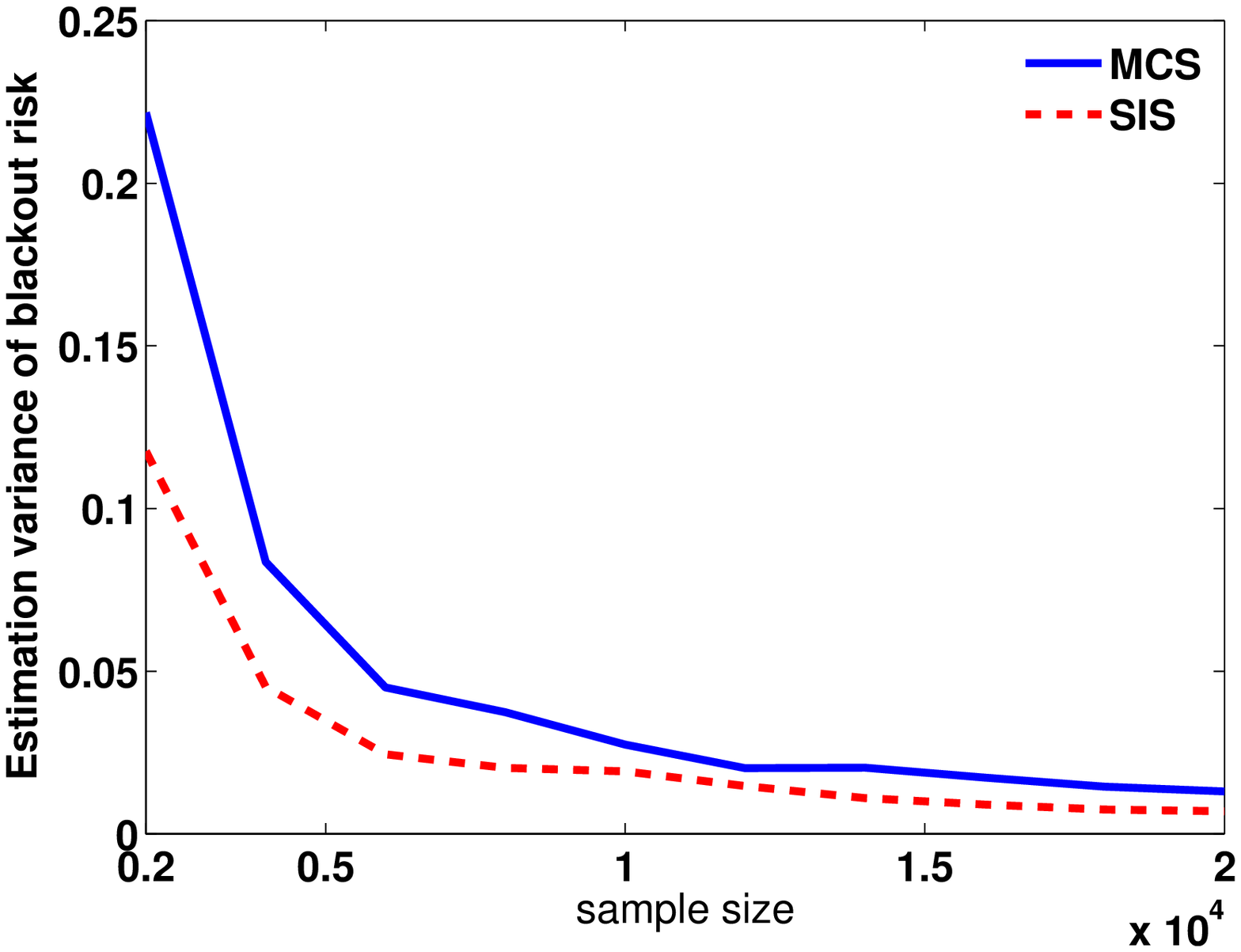}}
	\caption{Performance of  blackout risk estimation with different sample sizes}
	\label{fig.7} 
\end{figure}

\subsection{Case 2: A Real Power System}	



For further demonstrating the practicality of the SIS based strategy, we compare it with  the MCS strategy  in a large real power grid in China. The sample size is still set as 2,000. Similar to th previous case, both strategies can give unbiased estimation. Because of the space limitation, the results about the unbiasedness of estimation are omitted here, while the estimation variance of load shedding probability and blackout risk are shown in Tab. \ref{t1} and Tab. \ref{t2}, respectively.

In this case, the SIS outperforms the MCS again. As for small ${Y_0}$, the estimation variance of the SIS is smaller compared with the MCS. When ${Y_0}$ increases, the difference is getting more and more significant. When ${Y_0}$ is large enough, say, $4,000$ MW in this case, the MCS cannot obtain any effective samples to carry out statistic analysis on rare events, while the SIS is still effective for capturing those rare events . This case further exhibits the proposed SIS strategy can remarkably improve the efficiency and the reliability of  cascading outage analysis compared with the traditional MCS strategy,  especially when extremely rare blackouts are involved.

\begin{table}[!t]
	\caption{Estimation variance of load loss probability  ($\times 10^{-7}$)}
	\label{t1}
	\centering
\begin{tabular}{c|cccccc}
	\hline \hline
	${Y_0}$(MW)&1,000&2,000&3,000&4,000&5,000&6,000\\
	\hline
	MCS &$5.6$&$5.6{e^{ - 1}}$&$9.6{e^{ - 2}}$&-&-&-\\
	\hline
	SIS &$7.8 $&$4.1{e^{ - 1}}$&$4.8{e^{ - 3}}$&$4.3{e^{ - 3}}$&$2.8{e^{ -5}}$&$2.0{e^{ - 11}}$\\
	\hline \hline
\end{tabular}
\end{table}

\begin{table}[!t]
	\caption{Estimation variance of blackout risk with MCS and SIS}
	\label{t2}
	\centering
	\begin{tabular}{c|cccccc}
		\hline \hline
		${Y_0}$(MW)&1,000&2,000&3,000&4,000&5,000&6,000\\
		\hline
		MCS &$1.13$&$0.21$&$6.1{e^{ - 2}}$&-&-&-\\
		\hline
		SIS &$0.77$&$2.7{e^{ - 2}}$&$7.5{e^{ - 3}}$&$9.4{e^{ - 5}}$&$8.1{e^{ - 5}}$&$1.3{e^{ - 10}}$\\
		\hline \hline
	\end{tabular}
\end{table}

	\section{Conclusion}
	
	In this paper, we have formulated a cascading outage in power systems as a Markov chain with specific state space and transition probability, based on which we have further derived a sequential importance sampling  strategy for cascading outage simulations.  Theoretical analysis and case studies show that
	\begin{enumerate}
		\item 	The  Markov chain based formulation of cascading outages is well defined, which admits standard and strict stochastic analysis.  With the formulation, it is expected that more powerful analytic tools  can be applied.
		
		\item The SIS based simulation strategy can significantly enhance the computational efficiency and the estimation reliability of cascading outage analysis.
		
		\item The SIS based simulation strategy can dramatically improve the capability of capturing very rare events in cascading outage simulations.
		
	\end{enumerate}
	
	Whereas the Markov chain based formulation and the SIS based simulation strategy are derived for the cascading outage analysis in power systems in this paper, it could provide  a generic framework for cascading outage analysis of a broad class of complex networks. We believe  Our ongoing work is to quantitatively characterize the confidence bounds of the estimation results of SIS based simulation strategies.

\section*{Acknowledgment}
The authors would like to thank S. H. Low and L. Guo for very helpful discussions.



%

\bibliographystyle{IEEEtran}
\bibliography{ESORef0915}


%








\end{document}